\pdfoutput=1
\documentclass{TOOLS/llncs}
\usepackage{amsfonts,amsmath}
\usepackage{graphicx}
\usepackage{url}
\usepackage{verbatim}
\usepackage{color}
\definecolor{lgray}{gray}{0.92}
\definecolor{lblue}{rgb}{0.90,0.90,1.00}
\definecolor{lyellow}{rgb}{1.00,1.00,0.70}

\usepackage{listings}
\newenvironment{codex}{\small\verbatim}{\endverbatim\normalsize}

\lstnewenvironment{code}
    {\lstset{}%
      \csname lst@SetFirstLabel\endcsname}
    {\csname lst@SaveFirstLabel\endcsname}
\newtheorem{prop}{Proposition}

\newcommand{\BI}[0]{\begin{itemize}}
\newcommand{\EI}[0]{\end{itemize}}

\newcommand{\BE}[0]{\begin{enumerate}}
\newcommand{\EE}[0]{\end{enumerate}}

\newcommand{\BX}[0]{\begin{codex}}
\newcommand{\EX}[0]{\end{codex}}
\newcommand{\BV}[0]{\begin{verbatim}}
\newcommand{\EV}[0]{\end{verbatim}}

\def \bscale1 {0.25}
\def \bscale {0.25}
\def \N {\mathbb{N}}

\newcommand{\FIG}[4]{
\begin{figure}[htbp]
\centering
{\includegraphics[scale=#3]{figs/#4}}
\caption{#2}
\label{#1}
\end{figure}
}

\title{
   Tree-based Arithmetic and Compressed Representations of Giant Numbers
}

\author{Paul Tarau}
\institute{
   Department of Computer Science and Engineering\\
   University of North Texas\\
   {\em e-mail: tarau@cs.unt.edu}
}

\begin{document}
\maketitle
\date{}

\begin{abstract}
Can we do arithmetic in a completely different way, 
with a radically different data structure?  
Could this approach provide practical benefits,
like operations on giant numbers while 
having an average performance similar to traditional 
bitstring representations?

While  answering these questions positively, 
our tree based representation described in this paper
comes with a few extra benefits:
it compresses giant numbers such that, for instance,
the largest known prime number as well
as its related perfect number are represented as trees of small sizes.
The same also applies to Fermat numbers and important computations like
exponentiation of two become constant time operations. 

At the same time,
succinct representations of sparse sets, multisets and sequences become
possible through bijections to our tree-represented natural numbers.
\end{abstract}

\section{Introduction}
If extraterrestrials were to do arithmetic computations, would one assume that
their numbering system should look the same?  At a first thought, 
one might be inclined to think so, maybe with the exception of the actual 
symbols used to denote digits or the base of the system likely to match numbers of fingers or some other anthropocentric criteria. 
After some thinking about the endless diversity of the universe (or unconventional models of Peano's axioms), 
one might also consider the possibility that
the departure from our well-known number representations could be more significant.

With more terrestrial uses in mind, one would simply ask: is it possible to do arithmetic  differently, possibly including giant numbers and a radically different underlying data structure, while maintaining the same fundamentals -- addition, multiplication, exponentiation, primality, Peano's axioms behaving in the same way?
Moreover, can such exotic arithmetic be as efficient as binary arithmetic, 
while possibly supporting massively parallel execution?

We have shown in \cite{ppdp10tarau} that type classes
and polymorphism can be used to
share fundamental operations between natural numbers, 
finite sequences, sets and multisets.
As a consequence of this line of research, we have
discovered that it is possible to transport the recursion equations
describing binary number arithmetics on natural numbers
to various tree types.

However, the representation proposed in this paper is {\em different}. It is arguably 
simpler, more flexible and likely to support parallel execution of
operations.  It also allows to easily compute average and worse
case complexity bounds. We will describe it in full detail in the next
sections, but for the curious reader, it is essentially a {\em recursively self-similar 
run-length encoding of bijective base-2 digits}.

The paper is organized as follows. Section \ref{sharing} describes 
our type class used to share generic properties of natural numbers.
Section \ref{gen} describes binary arithmetic operations that
are specialized in section \ref{T} to our compressed tree representation.
Section \ref{theo} describes efficient tree-representations 
of some important number-theoretical entities like Mersenne, Fermat and
perfect numbers. 
Section \ref{sparse} shows that sparse sets, multisets and lists
have succinct tree representation.
Section \ref{iso} describes
generic isomorphisms between data types, centered
around transformations of instances of
our type class to their corresponding sets, multisets and lists.
Section \ref{red} shows interesting complexity reductions
in other computations and
section \ref{perfcomp} compares the performance of our tree-representation
with conventional ones.
Section \ref{related} discusses related work.
Section \ref{concl} concludes the paper and discusses future work.

To provide a concise view of our compressed tree data type and its operations, we will
use the strongly typed functional language Haskell as a precise means to
provide an {\em executable} specification. We have adopted
a literate programming style, i.e. the
code contained in the paper
forms a self-contained Haskell module (tested with ghc 7.4.1),
also available as a separate file at
\url{http://logic.cse.unt.edu/tarau/research/2013/giant.hs} .
Also, a Scala package implementing the same tree-based computation
is available from \url{http://code.google.com/p/giant-numbers/}.
We hope that this will encourage the reader to
experiment interactively and validate the
technical correctness of our claims.

We mention, for the benefit of the
reader unfamiliar with Haskell, that a notation like {\tt f x y} stands for $f(x,y)$,
{\tt [t]} represents sequences of type {\tt t} and a type declaration
like {\tt f :: s -> t -> u} stands for a function $f: s \times t \to u$
(modulo Haskell's ``currying'' operation, given the isomorphism between 
the function spaces ${s \times t} \to u$ and ${s \to t} \to u$). 
Our Haskell functions are always represented as sets
of recursive equations guided by pattern matching, conditional
to constraints (simple arithmetic relations following \verb~|~ and before
the \verb~=~ symbol).
Locally scoped helper functions are defined in Haskell
after the {\tt where} keyword, using the same equational style.
The composition of functions {\tt f} and {\tt g} is denoted {\tt f . g}.
It is also customary in Haskell, when defining functions in an equational style (using {\tt =})
to write $f=g$ instead of $f~x=g~x$ (``point-free'' notation).
The use of Haskell's ``call-by-need'' evaluation
allows us to work with infinite
sequences, like the {\tt [0..]} infinite list notation, corresponding to the
set $\N$ itself. 

Our literate Haskell program is organized as the module {\tt Giant}
depending on a few packages, as follows:
\begin{code}
{-# LANGUAGE NoMonomorphismRestriction #-}
-- needed to define toplevel isomorphisms
module Giant where
-- cabal install data-ordlist is required before importing this
import Data.List.Ordered
import Data.List hiding (unionBy)
import System.CPUTime
\end{code}
\begin{comment}
import Visuals
\end{comment}

\section{Sharing axiomatizations with type classes} \label{sharing}

Haskell's {\em type classes} \cite{DBLP:conf/popl/WadlerB89} are 
a good approximation of axiom
systems as they describe properties and operations 
generically i.e. in
terms of their action on objects of a parametric type. 
Haskell's type {\em instances}
approximate {\em interpretations} \cite{kaye07} of such axiomatizations by
providing implementations of the primitive operations, with the added benefit of
refining and possibly
overriding derived operations with more efficient equivalents.

We start by defining a type class that abstracts away 
properties of binary representations of natural numbers.

The class {\tt N} assumes only a theory of structural equality (as
implemented by the class {\tt Eq} in Haskell).
%and the {\tt Read/Show} superclasses needed for input/output. 
It implements a representation-independent
abstraction of natural numbers, allowing us to compare our
tree representation with ``ordinary'' natural numbers represented
as non-negative arbitrary large {\tt Integers} in Haskell,
as well as with a binary representation using bijective base-2
\cite{wiki:bijbase}.

\begin{code}
class Eq n => N n where
\end{code}
An instance of this class is required to implement the following 6
primitive operations:
\begin{code}
  e :: n
  o,o',i,i' :: n->n
  o_ :: n->Bool
\end{code}
The constant function {\tt e} can be seen as representing the empty sequence of 
binary digits. 
With the usual representation of natural numbers, {\tt e} will be interpreted as
{\tt 0}.  The constructors {\tt o} and {\tt i} can be seen as
applying a function that adds a {\tt 0} or {\tt 1} digit to a binary
string on {\tt \{0,1\}$^*$}. The deconstructors {\tt o'} and {\tt i'} undo
these operations by removing the corresponding digit. The recognizer
{\tt o\_} detects that the constructor {\tt o} is the last one applied,
i.e. that the``string ends with the {\tt 0} symbol. It will be interpreted
on $\mathbb{N}$ as a recognizer of odd numbers.
 
This type class also endows its instances with generic implementations of
the following derived operations:
\begin{code}
  e_,i_ :: n->Bool
  e_ x = x == e
  i_ x = not (e_ x || o_ x)
\end{code}
with structural equality used implicitly 
in the definition of the recognizer predicate for
empty sequences {\tt e\_}
and with the assumption that the domain is exhausted by the three recognizers
in the definition of the recognizer {\tt i\_} of
sequences ending with {\tt 1}, representing {\rm even} positive numbers
in bijective base 2.

Successor {\tt s} and predecessor {\tt s'} functions are implemented in terms
of these operations as follows:
\begin{code}  
  s,s' :: n->n
  
  s x | e_ x = o x
  s x | o_ x = i (o' x)
  s x | i_ x = o (s (i' x))
  
  s' x | x == o e = e
  s' x | i_ x = o (i' x)
  s' x | o_ x = i (s' (o' x))
\end{code}
By looking at the code, one might notice that our generic definitions
of operations mimic recognizers, constructors and destructors for bijective
base-2 numbers, i.e. sequences in the language $\{0,1\}^{*}$, similar to binary
numbers, except that 0 is represented as the empty sequence and left-delimiting by 1 is
omitted.

\begin{prop}
Assuming average constant time for recognizers, constructors and destructors {\tt e\_,o\_o,i\_,i,o',i'}, successor
and predecessor {\tt s} and {\tt s'} are constant time, on the average.
\end{prop}
\begin{proof}
Clearly the first two  rules are constant time for both {\tt s} and {\tt s'} as they do not make recursive calls. 
To show that the third rule applies recursion a
constant number of times on the average, we observe that the recursion steps are exactly given by the number of 0s or 1s
that a (bijective base-2 number) ends with. As only half of them end with a 0 and another half of those end with another 0
etc. one can see that the average number of 0s is bounded by ${1 \over 2}+{1 \over 4}+ \ldots=1$. The same reasoning
applies to the average number of 1s a number can end with.
\end{proof}
The infinite stream of generic natural numbers is generated 
by iterating over the successor operation {\tt s}:
\begin{code}
  allFrom :: n->[n]
  allFrom x = iterate s x
\end{code}

\section{Efficient arithmetic operations, generically} \label{gen}

We will first show that all fundamental arithmetic operations
can be described in this abstract, representation-independent framework. This
will make possible creating instances that, on top of
symbolic tree representations, provide implementations of these
operations with asymptotic efficiency comparable to the usual
bitstring operations.

We start with addition ({\tt add}) and subtraction ({\tt sub}):
\begin{code} 
  add :: n->n->n
  add x y | e_ x = y
  add x y | e_ y  = x
  add x y | o_ x && o_ y = i (add (o' x) (o' y))
  add x y | o_ x && i_ y = o (s (add (o' x) (i' y)))
  add x y | i_ x && o_ y = o (s (add (i' x) (o' y)))
  add x y | i_ x && i_ y = i (s (add (i' x) (i' y)))
  
  sub :: n->n->n
  sub x y | e_ y = x
  sub y x | o_ y && o_ x = s' (o (sub (o' y) (o' x))) 
  sub y x | o_ y && i_ x = s' (s' (o (sub (o' y) (i' x))))
  sub y x | i_ y && o_ x = o (sub (i' y) (o' x))  
  sub y x | i_ y && i_ x = s' (o (sub (i' y) (i' x))) 
\end{code}
It is easy to see that addition  and subtraction  are
implemented generically, with asymptotic complexity proportional 
to the size of the operands.
Comparison provides the expected total order of $\N$ on our type class:
\begin{code}
  cmp :: n->n->Ordering
  cmp x y | e_ x && e_ y = EQ
  cmp x _ | e_ x = LT
  cmp _ y | e_ y = GT
  cmp x y | o_ x && o_ y = cmp (o' x) (o' y)
  cmp x y | i_ x && i_ y = cmp (i' x) (i' y)
  cmp x y | o_ x && i_ y = down (cmp (o' x) (i' y)) where
    down EQ = LT
    down r = r
  cmp x y | i_ x && o_ y = up (cmp (i' x) (o' y)) where
    up EQ = GT
    up r = r
\end{code}
%\begin{comment}
And based on it one can define the minimum {\tt min2} and maximum 
{\tt max2} functions as follows:
\begin{code}
  min2,max2 :: n->n->n
  min2 x y = if LT==cmp x y then x else y
  max2 x y = if LT==cmp x y then y else x
\end{code}
%\end{comment}
Next, we define multiplication:
\begin{code}
  mul :: n->n->n
  mul x _ | e_ x = e
  mul _ y | e_ y = e
  mul x y = s (m (s' x) (s' y)) where
    m x y | e_ x = y
    m x y | o_ x = o (m (o' x) y)
    m x y | i_ x = s (add y  (o (m (i' x) y)))
\end{code}
as well as double of a number {\tt db} and %its left inverse 
half of an even number {\tt hf},
having both simple expressions:
\begin{code}  
  db,hf :: n->n 
  db = s' . o
  hf = s . i'
\end{code}
Power is defined as follows:
\begin{code}
  pow :: n->n->n
  pow _ y | e_ y = o e
  pow x y | o_ y = mul x (pow (mul x x) (o' y))
  pow x y | i_ y = mul (mul x x) (pow (mul x x) (i' y)) 
\end{code}
together with more efficient special instances,
exponent of 2 ({\tt exp2}) and multiplication by a power of 2
({\tt leftshift}):
\begin{code}
  exp2 :: n->n
  exp2 x | e_ x = o e
  exp2 x = db (exp2 (s' x)) 
  
  leftshift :: n->n->n
  leftshift x y = mul (exp2 x) y
\end{code}
Finally, division and reminder on $\N$ is a bit trickier but
can be expressed generically as well:
\begin{code}
  div_and_rem :: n->n->(n,n)
  
  div_and_rem x y | LT == cmp x y = (e,x)
  div_and_rem x y | not (e_ y) = (q,r) where 
    (qt,rm) = divstep x y
    (z,r) = div_and_rem rm y
    q = add (exp2 qt) z
    
    divstep :: N n => n->n->(n,n)
    divstep n m = (q, sub n p) where
      q = try_to_double n m e
      p = mul (exp2 q) m    
    
      try_to_double x y k = 
        if (LT==cmp x y) 
          then s' k
          else try_to_double x (db y) (s k)   
          
  divide,reminder :: n->n->n
  
  divide n m = fst (div_and_rem n m)
  reminder n m = snd (div_and_rem n m)
\end{code}

And for the reader curious by now about how this maps to ``arithmetic as
usual", here is an instance built around 
the (arbitrary length) {\tt Integer}
type, also usable as a witness on the time/space complexity 
of our operations.
\begin{code}
instance N Integer where
  e = 0
  
  o_ x = odd x
  
  o  x = 2*x+1
  o' x | odd x && x >  0 = (x-1) `div` 2
  
  i  x = 2*x+2
  i' x | even x && x > 0 = (x-2) `div` 2
\end{code}
An instance mapping our abstract operations to actual constructors,
follows in the form of the datatype {\tt B}
\begin{code}
data B = B | O B | I B deriving (Show, Read, Eq)

instance N B where
  e = B
  o = O
  i = I
  
  o' (O x) = x
  i' (I x) = x
  
  o_ (O _) = True
  o_ _ = False
  

\end{code}

One can try out various operations on these instances:
\begin{codex}
*Giant> mul 10 5
50
*Giant> exp2 5
32
*Giant> add (O B) (I (O B))
O (I B)
\end{codex}

\section{Computing with our compressed tree representations} \label{T}

We will now show how our shared axiomatization framework can be
implemented as a new, somewhat unusual instance, that brings
the ability to do arithmetic computations
with trees.

First, we define the data type for our tree represented natural numbers:
\begin{code}
data T = T | V T [T] | W T [T] deriving (Eq,Show,Read)
\end{code}
The intuition behind this ``union type'' is the following:
\begin{itemize}
\item The type {\tt T} corresponds to an empty sequence 
\item the type {\tt V x xs} counts the number {\tt x} of {\tt o} applications followed by an alternation
of similar counts of {\tt i} and {\tt o} applications

\item the type {\tt W x xs} counts the number {\tt x} of {\tt i} applications followed by an alternation
of similar counts of {\tt o} and {\tt i} applications

\item the same principle is applied recursively for the counters, until the empty sequence is reached
\end{itemize}
One can see this process as run-length compressed bijective base-2 numbers, represented as trees
with either empty leaves or at least one branch, after applying the encoding recursively.
First we define the 6 primitive operations:
\begin{code}
instance N T where
  e = T
  
  o T = V T []
  o (V x xs) = V (s x) xs
  o (W x xs) = V T (x:xs)

  i T = W T []
  i (V x xs) = W T (x:xs)
  i (W x xs) = W (s x) xs
  
  o' (V T []) = T
  o' (V T (x:xs)) = W x xs
  o' (V x xs) = V (s' x) xs

  i' (W T []) = T
  i' (W T (x:xs)) = V x xs
  i' (W x xs) = W (s' x) xs
  
  o_ (V _ _ ) = True
  o_ _ = False
\end{code}
Next, we override two operations involving exponents of  {\tt 2} as follows
\begin{code}
  exp2 = exp2' where
    exp2' T = V T []
    exp2' x = s (V (s' x) [])

  leftshift = leftshift' where
    leftshift' _ T = T
    leftshift' n y | o_ y = s (vmul n (s' y))
    leftshift' n y | i_ y = s (vmul (s n) (s' y))
\end{code}
The {\tt leftshift'} operation uses
an efficient implementation, specialized for the type {\tt T}, of the
repeated application (n times) of constructor {\tt o}, over the second argument
of the function {\tt vmul}:
% same as otimes !!!
\begin{code}
vmul T y = y
vmul n T = V (s' n) []
vmul n (V y ys) = V (add (s' n) y) ys
vmul n (W y ys) = V (s' n) (y:ys)
\end{code}

Note that such overridings take advantage of the specific encoding, as a result
of simple number theoretic observations. For instance, the operation
{\tt exp2'} works in time proportional to {\tt s} and {\tt s'}, that can be shown to
be constant on the average. The more complex {\tt leftshift'} operation observes
that repeated application of the {\tt o} operation, when adjusted 
based on being even or odd, provides an efficient implementation of multiplication
with an exponent of 2.

It is convenient at this point, as we target a diversity of interpretations
materialized as Haskell instances, to provide a polymorphic converter between
two different instances of the type class {\tt N} as well as their
associated lists, implemented by structural recursion over the representation
to convert. The function {\tt view} allows importing a wrapped
object of a different instance of {\tt N}, generically.
\begin{code}
view :: (N a,N b)=>a->b
view x | e_ x = e
view x | o_ x = o (view (o' x))
view x | i_ x = i (view (i' x))
\end{code}

We can specialize {\tt view} to provide conversions to
our three data types, each denoted with the corresponding
lower case letter, {tt t}, {\tt b} and {\tt n} for the usual natural numbers.
\begin{code}
t :: (N n) => n -> T
t = view

b :: (N n) => n -> B
b = view

n :: (N n) => n -> Integer
n = view
\end{code}
One can try them out as follows:
\begin{codex}
*Giant> t 42
W (V T []) [T,T,T]
*Giant> b it
I (I (O (I (O B))))
*Giant> n it
42
\end{codex}

While space constraints forbid us from providing the correctness
proofs of operations like {\tt exp2'} and {\tt leftshift'}, we are
able to illustrate their expected usage
as follows:
\begin{codex}
*Giant> t 5
V T [T]
*Giant> exp2 it
W T [V (V T []) []]
*Giant> n it
32
*Giant> t 10
W (V T []) [T]
*Giant> leftshift it (t 1)
W T [W T [V T []]]
*Giant> n it
1024

\end{codex}

\section{Efficient representation of some important number-theoretical entities} \label{theo}

%\subsection{Compact representation of large Mersenne, Fermat and perfect numbers}

Let's first observe that Fermat, Mersenne and perfect numbers have all compact
expressions with our tree representation of type {\tt T}.

\begin{code}
fermat n = s (exp2 (exp2 n))

mersenne p = s' (exp2 p)

perfect p = s (V q [q]) where q = s' (s' p)
\end{code}

And one can also observe that this contrasts with both the {\tt Integer} representation and
the bijective base-2 numbers {\tt B}:

\begin{codex}
*Giant> mersenne (b 127)
O (O (O (O (O (O (O (O (O (O (O (O (O (O (O 
  ... a few lines of Os and Is
  )))))))))))))))))))))))))))))))))))))))))
  *Giant> mersenne (n 127)
170141183460469231731687303715884105727
*Giant> mersenne (t 127)
V (W (V T [T]) []) []
\end{codex}

The largest known prime number, found by the GIMP distributed computing 
project \cite{wiki:gimp} is the 45-th Mersenne prime = $2^{43112609}-1$.
It is defined in Haskell as follows:
\begin{code}
-- its exponent
prime45 = 43112609 :: Integer

-- the actual Mersenne prime
mersenne45 = s' (exp2 (t p)) where 
  p = prime45::Integer

\end{code}

While it has a bitsize of 43112609, we have observed that its 
compressed tree representation using our type {\tt T} is rather small:
\begin{codex}
*Giant> mersenne45
V (W T [V (V T []) [],T,T,T,W T [],V T [],T,W T [],W T [],T,V T [],T,T]) []
\end{codex}
One the other hand, displaying it with a decimal or binary representation
would take millions of digits.

And by folding replicated subtrees to obtain an equivalent DAG representation,
one can save even more memory. Figure \ref{mersenne45} shows this representation,
involving only 6 nodes.

\FIG{mersenne45}{Largest known prime number: the 45-th Mersenne prime, represented as a DAG}{0.40}{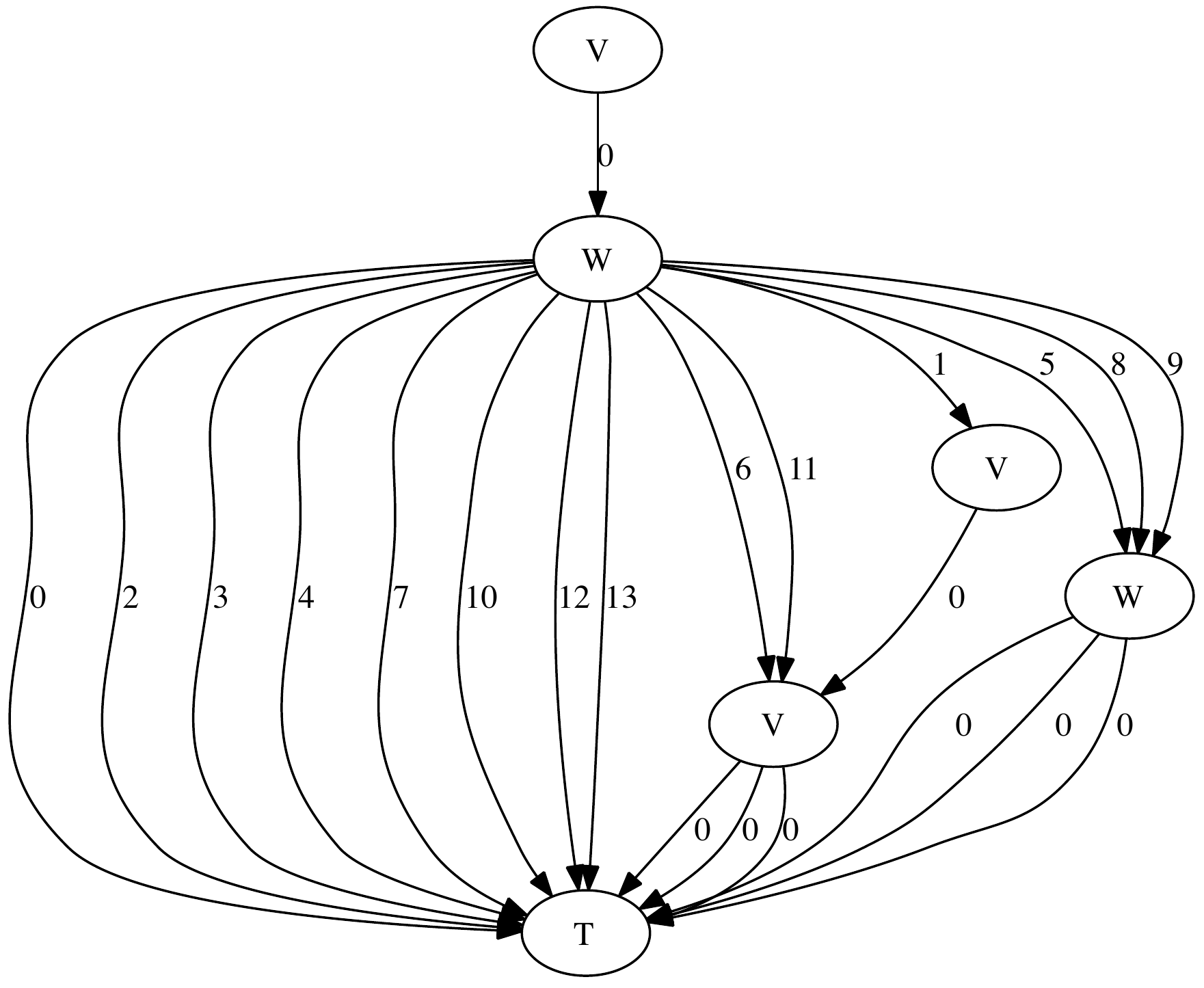}

It is interesting to note that similar compact representations can also
be derived for perfect numbers. For instance, the largest known perfect number, derived
from the largest known Mersenne prime as $2^{43112609-1}(2^{43112609}-1)$,
is:
\begin{code}
perfect45 = perfect (t prime45)
\end{code}
Fig. \ref{perfect45} shows its DAG representation involving only 7 nodes.
\FIG{perfect45}{Largest known perfect number}{0.40}{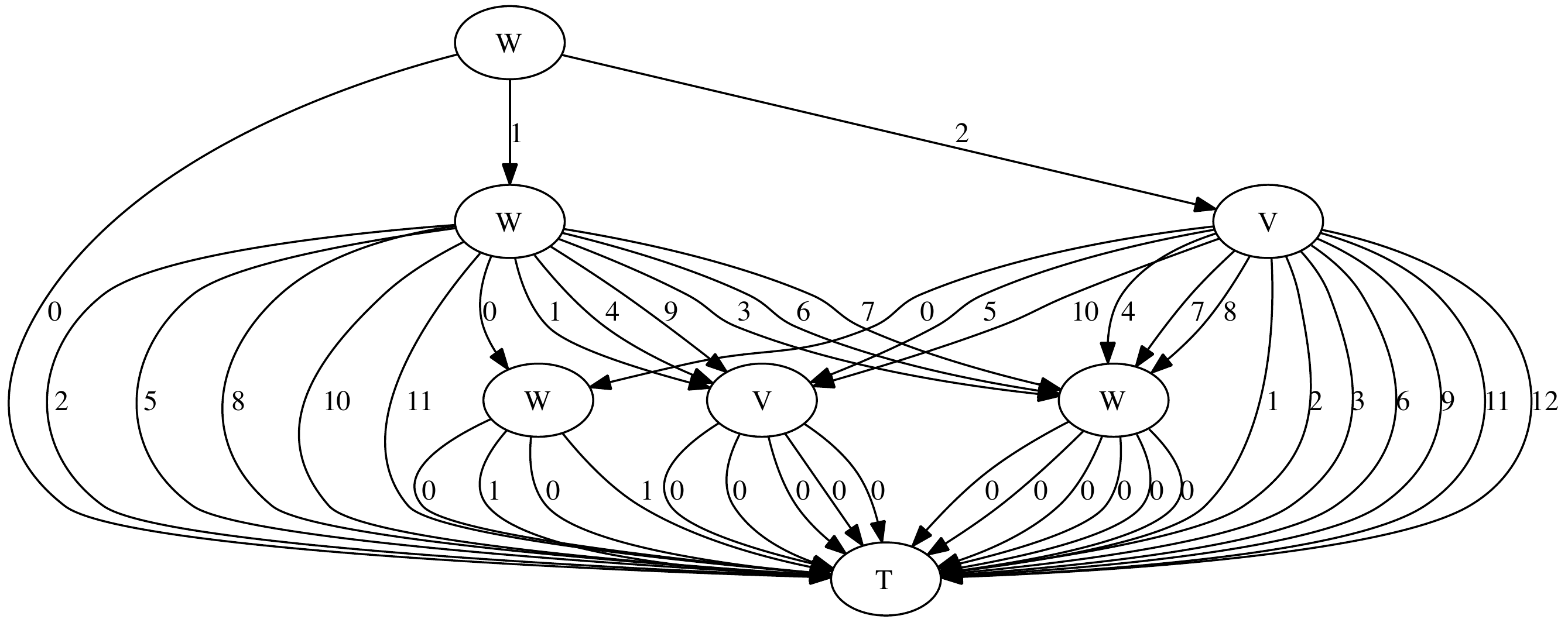}
Similarly, the largest Fermat number that has been factored so far, F11=$2^{2^{11}}+1$ is compactly represented as 
\begin{codex}
*Giant> fermat (t 11)
V T [T,V T [W T [V T []]]]
\end{codex}
By contrast, its (bijective base-2) binary representation consists of 2048 digits.

\section{Representing sparse sets, multisets and lists} \label{sparse}

We will now describe bijective mappings between collection types as
well as a G\"odel numbering scheme putting them in bijection with
natural numbers. Interestingly, natural number encodings for
sparse instances of these collections will have space-efficient
representations as natural numbers of type {\tt T}, in contrast
with bitstring or conventional {\tt Integer}-based representations.

The type class {\tt Collections} will convert between natural numbers and lists,
by using the bijection $f(x,y)=2^{x}(2y+1)$, implemented by the function {\tt c}
and its first and second projections
{\tt c'} and {\tt c''}, inverting it.
\begin{code}
class (N n) => Collections n where
  c :: n->n->n
  c',c'' :: n->n
 
  c x y = mul (exp2 x) (o y)
  
  c' x | not (e_ x) = if o_ x then e else s (c'  (hf x))
  c'' x | not (e_ x) = if o_ x then o' x else c'' (hf x)
\end{code}
The bijection between natural numbers and lists of natural numbers, {\tt to\_list}
and its inverse {\tt from\_list} apply repeatedly {\tt c}
and respectively {\tt c'} and {\tt c''}.

\begin{code}
  to_list :: n->[n]
  to_list x | e_ x = []
  to_list x = (c' x) : (to_list (c'' x))

  from_list:: [n]->n
  from_list [] = e
  from_list (x:xs) = c x (from_list xs)
\end{code}
Incremental sums are used to transform arbitrary
lists to multisets and sets, inverted by pairwise differences.
\begin{code}  
  list2mset, mset2list, list2set, set2list :: [n]->[n]
 
  list2mset ns = tail (scanl add e ns)
  mset2list ms =  zipWith sub ms (e:ms)
    
  list2set = (map s') . list2mset . (map s)
  set2list = (map s') . mset2list . (map s) 
\end{code}
By composing with natural number-to-list bijections, we obtain bijections
to multisets and sets.
\begin{code}  
  to_mset, to_set :: n->[n]
  from_mset, from_set :: [n]->n

  to_mset = list2mset . to_list
  from_mset = from_list . mset2list
  
  to_set = list2set . to_list
  from_set = from_list . set2list
\end{code}
We will add the usual instances to the type class Collections. 
A simple number-theoretic observation connecting $2^{n}$ and $n$ 
applications of the constructor {\tt i}, implemented by the 
function {\tt vmul}, allows a shortcut that speeds up the 
bijection from lists to natural numbers, by overriding the 
functions {\tt c, c', c''} in instance {\tt T}. 
\begin{code}
instance Collections B
instance Collections Integer
instance Collections T where
  c = cons where
    cons n y = s (vmul n (s' (o y)))
    
  c' = hd where  
    hd z | o_ z = T
    hd z = s x where
      V x _ = s' z
      
  c'' = tl where
    tl z | o_ z = o' z 
    tl z = f xs where
      V _ xs = s' z
 
      f [] = T
      f (y:ys) = s (i' (W y ys))
 
\end{code}
As the following example shows, trees of type {\tt T} offer a significantly more compact
representation of sparse sets.
\begin{codex}
*Giant> from_set (map t [1,100,123,234])
W (V T []) [V T [T,W T [],T],T,V T [V T [],T],T,V T [W T [],T,T]]
*Giant> from_set (map n [1,100,123,234])
27606985387162255149739023449108112443629669818608757680508075841159170
*Giant> from_set (map b [1,100,123,234])
I (I (O (O (O (O (O (O (O (O (O (O (O (O (O (O (O (O (O ...
... a few lines ...
..)))))))))))
\end{codex}
Note that a similar compression occurs for sets of natural numbers with only a few elements missing, as
they have the same representation size with type {\tt T} as the {\tt dual} of their sparse counterpart.
\begin{codex}
Giant> from_set ([1,3,5]++[6..220])
3369993333393829974333376885877453834204643052817571560137951281130
*Giant> t it
W (V T []) [T,T,T,W (W T []) [T,T,T,T]]
*Giant> dual it
V (V T []) [T,T,T,W (W T []) [T,T,T,T]]
*Giant> to_set it
[T,V T [],W T [T],W T [T,W T [],T,T]]
*Giant> map n it
[0,1,4,220]
\end{codex}

As an application, we can define bitwise operations
on our natural numbers by borrowing the
corresponding ordered set operations, provided in Haskell
by the package {\tt Data.List.Ordered}.

First we define the type class {\tt BitwiseOperations}
and the higher order function {\tt l\_op} transporting
a binary operation from ordered sets to natural numbers.
\begin{code}
class Collections n => BitwiseOperations n where 
  l_op :: ([n]->[n]->[n])->n->n->n
  l_op op x y = from_set (op (to_set x) (to_set y))
\end{code}
Next we define the bitwise {\tt and}, {\tt or} and {\tt xor} operations:
\begin{code}  
  l_and,l_or,l_xor,l_dif:: n->n->n
  l_and = l_op (isectBy cmp)
  l_or = l_op (unionBy cmp)
  l_xor = l_op (xunionBy cmp)
  l_dif = l_op (minusBy cmp)
\end{code}
More complex operations like the ternary {\tt if-the-else}
can be defined as a combination of binary operations:
\begin{code}  
  l_ite :: n->n->n->n
  l_ite x y z = from_set (ite (to_set x) (to_set y) (to_set z)) where
    ite d a b =  xunionBy cmp e b where
      c = xunionBy cmp a b
      e = isectBy cmp c d
\end{code}
Finally, bitwise negation (requiring additional parameter {\tt l}
defining the bitlength of the operand) can be defined using
set complement with respect to $2^{l}-1$, corresponding to the set of all
elements up to {\tt l}.
\begin{code}   
  l_not :: Integer->n->n
  l_not l x |xl<=l = from_set (minusBy cmp ms xs) where
    xs = to_set x
    xl = genericLength xs
    ms = genericTake l (allFrom e)
\end{code}
We will next generalize the iso-functor mechanism 
implemented generically by {\tt l\_op} that
transports operations back an forth
between data types as a special data type,
consisting of two higher order functions
inverse to each other.

\section{Isomorphisms between data types, generically} \label{iso}
Along the lines of \cite{everything} we can define
isomorphisms between data types as follows:
\begin{code}
data Iso a b = Iso (a->b) (b->a)

from (Iso f _) = f
to (Iso _ f') = f'
\end{code}
Morphing between data types as well as ``lending'' operations from
one to another is provided by the combinators {\tt as}, {\tt land1} and {\tt land2}:
\begin{code}
as that this x = to that (from this x)

lend1 op1 (Iso f f') x = f' (op1 (f x))
lend2 op2 (Iso f f') x y = f' (op2 (f x) (f y))
\end{code}
Assuming that the Haskell option {\tt NoMonomorphismRestriction} is set on, we can
now define generically ``virtual types'' centered around type class {\tt N}:
\begin{code}
nat = Iso id id
list = Iso from_list to_list
mset = Iso from_mset to_mset
set = Iso from_set to_set
\end{code}
This results in a small ``embedded language'' that morphs between
various instances of type class {\tt N} and their corresponding
list, multiset and set types, as follows:
\begin{codex}
*Giant> as set nat 1234
[1,4,6,7,10]
*Giant> map t it
[V T [],W T [T],W (V T []) [],V (W T []) [],W (V T []) [T]]
*Giant> as nat set it
W (V T []) [V T [],T,T,V T [],V T []]
*Giant> n it
1234
*Giant> lend1 s set [0,2,3]
[1,2,3]
*Giant> lend2 add set [0,2,3] [4,5]
[0,2,3,4,5]
\end{codex}

\section{Complexity reduction in other computations} \label{red}

A number of other, somewhat more common computations also benefit from our data representations.
The type class {\tt SpecialComputations} groups them together and provides their bitstring
inspired generic implementations.

The function {\tt dual} flips {\tt o} and {\tt i} operations
for a natural number seen as written in bijective base 2. 
\begin{code}
class Collections n => SpecialComputations n where
  dual :: n->n
  dual x | e_ x = e
  dual x | o_ x = i (dual (o' x))
  dual x | i_ x = o (dual (i' x))
\end{code}
The function {\tt bitsize} computes  the number of applications of the 
{\tt o} and {\tt i} operations:
\begin{code}    
  bitsize :: n->n
  bitsize x | e_ x = e
  bitsize x | o_ x = s (bitsize (o' x))
  bitsize x | i_ x = s (bitsize (i' x))
\end{code}
The function {\tt repsize} computes the representation size, which defaults to
the bit-size in bijective base 2: 
\begin{code}  
  -- representation size - defaults to bitsize
  repsize :: n->n
  repsize = bitsize
\end{code}
The functions {\tt decons} and {\tt cons} provide bijections
between $\N-\{0\}$ and $\N \times \N$ and can be used as
an alternative mechanism for building bijections
between lists, multisets and sets of natural numbers and natural numbers.
They are based on separating {\tt o} and {\tt i} applications
that build up a natural number represented in bijective base 2.
\begin{code}
  decons ::n->(n,n)
  cons :: (n,n)->n
   
  decons z | o_ z = (x,y) where
    x0 = s' (ocount z)
    y = otrim z
    x = if e_ y then (s'.o) x0 else x0 
  decons z | i_ z = (x,y) where
    x0 = s' (icount z)
    y = itrim z
    x = if e_ y then (s'.i) x0 else x0 
    
  cons (x,y) | e_ x && e_ y = s e
  cons (x,y) | o_ x && e_ y = itimes (s (i' (s x))) e
  cons (x,y) | i_ x && e_ y = otimes (s (o' (s x))) e
  cons (x,y) | o_ y = itimes (s x) y
  cons (x,y) | i_ y = otimes (s x) y
\end{code}
Implementing {\tt decons} and {\tt cons} requires
counting the number of applications of {\tt o} and {\tt i}
provided by {\tt ocount} and {\tt icount}, as well trimming
the applications of {\tt o} and {\tt i}, performed
by {\tt otrim} and {\tt itrim}.
\begin{code}    
  ocount,icount,otrim,itrim :: n->n
  
  ocount x | o_ x = s (ocount (o' x))
  ocount _ = e
  
  icount x | i_ x = s (icount (i' x))
  icount _ = e
  
  otrim x | o_ x = otrim (o' x)
  otrim x = x
  
  itrim x | i_ x = itrim (i' x)
  itrim x = x
\end{code}

\begin{code}
  otimes,itimes :: n->n->n
  
  otimes x y | e_ x = y
  otimes x y = otimes (s' x) (o y)
   
  itimes x y | e_ x = y
  itimes x y = itimes (s' x) (i y)
\end{code}

An alternative bijection between natural numbers and 
lists of natural numbers, {\tt to\_list'}
and its inverse {\tt from\_list'} is obtained by 
applying repeatedly {\tt cons}
and respectively {\tt decons}.

\begin{code}
  to_list' :: n->[n]
  to_list' x | e_ x = []
  to_list' x = hd : (to_list' tl) where (hd,tl)=decons x

  from_list' :: [n]->n
  from_list' [] = e
  from_list' (x:xs) = cons (x,from_list' xs)
\end{code}

One can observe the significant reduction of asymptotic complexity
with respect to the default operations provided by the type
class {\tt SpecialComputations}
when overriding with {\tt tbitsize} and {\tt tdual} in instance {\tt T}.
\begin{code}  
instance SpecialComputations Integer
instance SpecialComputations B
instance SpecialComputations T where
  bitsize = tbitsize where
    tbitsize T = T
    tbitsize (V x xs) = s (foldr add1 x xs)
    tbitsize (W x xs) = s (foldr add1 x xs)
    
    add1 x y = s (add x y)
    
  dual = tdual where
    tdual T = T
    tdual (V x xs) = W x xs
    tdual (W x xs) = V x xs 
    
  repsize = tsize where
    tsize T = T
    tsize (V x xs) = s (foldr add T (map tsize (x:xs)))
    tsize (W x xs) = s (foldr add T (map tsize (x:xs)))
\end{code}
The replacement with special purpose code for the {\tt cons / decons} functions
is even more significant:
\begin{code}
  decons (V x []) = ((s'.o) x,T)
  decons (V x (y:ys)) = (x,W y ys)
  decons (W x []) = ((s'.i) x,T)
  decons (W x (y:ys)) = (x,V y ys)

  cons (T,T) = V T []
  cons (x,T) | o_ x =  W (i' (s x)) []
  cons (x,T) | i_ x = V (o' (s x)) []
  cons (x,V y ys) = W x (y:ys)
  cons (x,W y ys) = V x (y:ys)
\end{code}
One can also see that their complexity is now proportional to {\tt s} and {\tt s'}
given that the {\tt V} and {\tt W} operations perform in constant time the 
work of {\tt otimes} and {\tt itimes}.
The following example illustrates their work:
\begin{codex}
*Giant> map to_list' [0..20]
[[],[0],[1],[2],[0,0],[0,1],[3],[4],[0,2],[0,0,0],[1,0],[1,1],
    [0,0,1],[0,3],[5],[6],[0,4],[0,0,2],[1,2],[1,0,0],[0,0,0,0]]
*Giant> map from_list' it
[0,1,2,3,4,5,6,7,8,9,10,11,12,13,14,15,16,17,18,19,20]\end{codex}
Note the shorter lists created close to powers of 2, coming from
the longer sequences of consecutive 
{\tt o} and {\tt i}m operations in that region.

\section{A performance comparison} \label{perfcomp}

Our performance measurements (run on a Mac Air with 8GB of memory and an Intel i7 processor) serve two objectives:
\begin{enumerate}
\item to show that, on the average, our tree based representations perform on a blend of arithmetic computations within a small constant factor compared with conventional bitstring-based computations

\item to show that on interesting special computations they outperform the conventional ones due to the much lower asymptotic complexity of such operations on data type {\tt T}.
\end{enumerate}

\begin{figure} 
\begin{center}
\begin{tabular}{||l||r|r|r|r|r||}
\hline
\hline
\multicolumn{1}{||c||}{{\small Benchmark}} & 
  \multicolumn{1}{c|}{{\small Integer}} &
  \multicolumn{1}{c|}{{\small binary type {\tt B}}} &
  \multicolumn{1}{c|}{{\small tree type {\tt T}}}\\
\hline 
\hline
 {\small Ackermann 3 7}   & 9418 & 7392 & 12313   \\ \hline
 {\small exp2 (exp2 14)}   & 23 & 315 & 0   \\ \hline
 {\small sparse set encoding} & 7560 & 2979 & 45     \\ \hline
 
 {\small bitsize of Mersenne 45} & ? & ? & 0    \\ \hline
 {\small bitsize of Perfect 45}   & ? & ? & 2   \\ \hline \hline

{\small generating primes} & 2722 & 2567 & 3591    \\ \hline
 {\small Mersenne prime tests} &  6925  & 6431 & 15037   \\ \hline \hline

\hline
\hline
\end{tabular} \\
\medskip
\caption{Time (in ms.) on a few benchmarks
\label{perf}}
\end{center}
\end{figure}

Objective {\em 1} is served by the Ackerman function that exercises the successor and predecessor functions quite heavily, the prime generation and the Lucas-Lehmer primality test for Mersenne numbers that exercise a blend of arithmetic operations.

Objective {\em 2} is served by the other benchmarks that take advantage of the overriding by instance {\tt T} of operations like {\tt exp2} and {\tt bitsize}, as well as the  compressed representation of large numbers like the 45-th Mersenne prime and perfect numbers. In some cases the conventional representations are unable to run these benchmarks within existing computer memory and CPU-power limitations (marked with {\tt ?} in the comparison table of Fig. \ref{perf}). In other cases, like in the sparse set encoding benchmark, data type {\tt T} performs significantly faster than binary representations.

Together they indicate that our tree-based representations are likely to be competitive with existing bitstring-based packages on typical computations and significantly outperform them on some number-theoretically interesting computations. While the code of the benchmarks is omitted due to space constraints, it is part of the companion Haskell file at \url{http://logic.cse.unt.edu/tarau/Research/2013/giant.hs}.

\section{Related work} \label{related}

We will briefly describe here some related work
that has inspired and facilitated this line of research
and will help to put our past contributions and planned
developments in context.

Several notations for very large numbers have been invented in the past. Examples
include Knuth's {\em arrow-up} notation \cite{knuthUp}
covering operations like the {\em tetration} (a notation for towers of exponents).
In contrast to our tree-based natural numbers,
such notations are not closed under
addition and multiplication, and consequently
they cannot be used as a replacement
for ordinary binary or decimal numbers.

The first instance of a hereditary number system, at our best knowledge,
occurs in the proof of Goodstein's theorem \cite{goodstein}, where
replacement of finite numbers on a tree's branches by the ordinal $\omega$
allows him to prove that a ``hailstone sequence'' visiting arbitrarily
large numbers eventually turns around and terminates.

Numeration systems on regular languages have been studied
recently, e.g. in \cite{Rigo2001469} and specific instances
of them are also known as 
bijective base-k numbers \cite{wiki:bijbase}.
Arithmetic packages similar to 
our bijective base-2 view of arithmetic operations
are part of libraries of proof assistants
like Coq \cite{Coq:manual} and the
corresponding regular 
languages have been used as a basis
of decidable arithmetic systems like
{\tt (W)S1S} \cite{buchi62} and {\tt (W)S2S} 
\cite{Rabin69Decidability}.

Arithmetic computations based
on recursive data types like
the free magma of binary trees
(isomorphic to the 
context-free language of balanced parentheses)
are described in \cite{sac12} and \cite{padl12},
where they are seen as G\"odel's {\tt System T} types,
as well as combinator application trees.
In \cite{ppdp10tarau}
a type class mechanism is used
to express computations on hereditarily
finite sets and hereditarily finite
functions.

An emulation of Peano and conventional binary arithmetic operations
in Prolog, is described in \cite{DBLP:conf/flops/KiselyovBFS08}.
Their approach is similar as far as a symbolic representation is used.
The key difference with our work is that our operations work on tree
structures, and as such, they are not based on previously known algorithms.

Arithmetic computations with types expressed as {\tt C++} templates
are described in \cite{DBLP:journals/corr/cs-CL-0104010} and in
online articles by Oleg Kiselyov using Haskell's
type inference mechanism. However, the algorithm described there is
basically the same as \cite{DBLP:conf/flops/KiselyovBFS08}, focusing
on Peano and binary arithmetics.

Efficient representation of sparse sets are usually based on a dedicated
data structure \cite{Briggs93anefficient} and they cannot be 
at the same time used for arithmetic computations as it is 
the case with our tree-based encoding.

In \cite{vu09} integer decision diagrams are introduced
providing a compressed representation for sparse
integers, sets and various other data types. 
However likewise \cite{sac12} 
and \cite{ppdp10tarau}, and by contrast to those
proposed in this paper, they do not compress
dense sets or numbers.

{\em Ranking} functions (bijections between data types and natural numbers)
can be traced back to G\"{o}del numberings
\cite{Goedel:31} associated to formulae. 
Together with their inverse {\em unranking} functions they are also 
used in combinatorial generation
algorithms
\cite{conf/mfcs/MartinezM03,knuth06draft}.

As a fresh look at the topic, we mention 
recent work in the context of
functional programming on
connecting heterogeneous data types through
bijective mappings and natural number 
encodings \cite{everybit,complambda,calc10tarau}.

\section{Conclusion and future work} \label{concl}

We have seen that
the average performance of
arithmetic computations with  trees of type {\tt T}
is comparable, up to small constant
factors, to computations performed
with the binary data type {\tt B}
and it outperforms them by an arbitrarily large margin on the
interesting special cases favoring the tree representations.

Still, does that
mean that such binary trees can be used as a basis
for a practical arbitrary integers package?

Native arbitrary length integer libraries like
GMP or BigInteger
take advantage of fast arithmetic on 64 bit words.

To match their performance, we plan to
switch between
bitstring representations for
numbers fitting in a machine word
and a a
tree representation for numbers not
fitting in a machine word. 

We have shown that some interesting number-theoretical entities like
Mersenne, Fermat and perfect numbers have significantly more compact
representations with our tree-based numbers. One may observe their
common feature: they are all represented in terms of exponents of 2,
successor/predecessor and specialized multiplication operations.

The fundamental theoretical challenge raised at this point is the following:
{\em can other number-theoretically interesting operations, with possible
applications to cryptography be also
expressed succinctly in terms of our tree-based data type? Is it possible to reduce the
complexity of some other important operations, besides those found so far?}

The methodology to be used relies on two key components, that
have been proven successful so far, 
in discovering succinct representations for Mersenne, Fermat and perfect numbers,
as well as low complexity algorithms for operations like {\tt bitsize} and {\tt exp2}:
\begin{itemize}
\item partial evaluation  of functional programs with respect to known data types and operations on them,
 as well as the use of other program transformations
\item salient number-theoretical observations, provable by induction, 
that relate operations on our tree data types to known identities and number-theoretical algorithms
\end{itemize}

\bibliographystyle{INCLUDES/splncs}
\bibliography{INCLUDES/theory,tarau,INCLUDES/proglang,INCLUDES/biblio,INCLUDES/syn}

\section*{Appendix}
This appendix contains some additional code, used for testing
and benchmarking our functions, grouped in the type class {\tt Benchmarks}.
First we define a prime number generator working with all our instances.
\begin{code}
class SpecialComputations n => Benchmarks n where
  primes :: [n]

  primes = s (s e) : filter is_prime (odds_from (s (s (s e)))) where
    odds_from x = x : odds_from (s (s x))
    
    is_prime p = [p]==to_primes p
    
    to_primes n = to_factors n p ps where
       (p:ps) = primes
    
    to_factors n p ps | cmp (mul p p) n == GT = [n]
    to_factors n p ps | e_ r =  p : to_factors q p ps where
       (q,r) = div_and_rem n p
    to_factors n p (hd:tl) = to_factors n hd tl
\end{code}
Next we define the Lucas-Lehmer fast primality test for Mersenne numbers:
\begin{code}
  lucas_lehmer :: n -> Bool
  lucas_lehmer p = e_ y where
    p_2 = s' (s' p)
    four = i (o e)
    m  = exp2 p
    m' = s' m
    y = f p_2 four
 
    f k n | e_ k = n
    f k n = r where 
      x = f (s' k) n
      y = s' (s' (mul x x))
      --r = reminder y m'
      r = fastmod y m 

    -- fast computation of k mod 2^p-1  
    fastmod k m | k == s' m = e
    fastmod k m | LT == cmp k m = k
    fastmod k m = fastmod (add q r) m where
       (q,r) = div_and_rem k m
  
  -- exponents leading to Mersenne primes
  mersenne_prime_exps :: [n]
  mersenne_prime_exps = filter lucas_lehmer primes
  
  -- actual Mersenne primes
  mersenne_primes :: [n]
  mersenne_primes = map f mersenne_prime_exps where
    f p = s' (exp2 p)   
\end{code}
The Ackerman function is a good benchmark for
successor and predecessor operations:
\begin{code}    
  ack :: n->n->n
  ack x n | e_ x = s n
  ack m1 x | e_ x = ack (s' m1) (s e)
  ack m1 n1 = ack (s' m1) (ack m1 (s' n1))
\end{code}
Next we define a variant of the 3x+1 problem / Collatz conjecture / Syracuse
function %\cite{arxiv:lagarias} 
(see \url{http://en.wikipedia.org/wiki/Collatz_conjecture})
that, somewhat surprisingly, can be expressed as
a mix of arithmetic operations and reflected list operations, to test
the relative performance of some of our instances. It is easy to show
that the Collatz conjecture is true iff the function {\tt nsyr}, implementing the
n-th iterate of the {\em Syracuse function}, always
terminates:
\begin{code}
  syracuse  :: n->n
  -- n->c'' (3n+2)
  syracuse n = c'' (add n (i n))

  nsyr :: n->[n]
  nsyr n | e_ n = [e]
  nsyr n = n : nsyr (syracuse n)
\end{code}
Finally we close our type class with the usual instance declarations:
\begin{code}
instance Benchmarks Integer
instance Benchmarks B
instance Benchmarks T
\end{code}
The following example illustrates
the first 8 sequences of the Syracuse function:
\begin{codex}
*Giant> map nsyr [0..7]
[[0],[1,2,0],[2,0],[3,5,8,6,2,0],[4,3,5,8,6,2,0],
     [5,8,6,2,0],[6,2,0],[7,11,17,26,2,0]]
\end{codex}

\begin{code}
\end{code}

Our generic benchmark function measures the CPU time
for running a no argument toplevel function {\tt f} received
as a parameter.
\begin{code}
benchmark mes f = do
  x<-getCPUTime
  print f
  y<-getCPUTime
  let time=(y-x) `div` 1000000000
  return (mes++" :time="++(show time))
\end{code}
The following benchmarks provide the code used
in the section \ref{perfcomp}.
\begin{code}
bm1t = benchmark "ack 3 7 on t" (ack (t (toInteger 3)) (t (toInteger 7)))
bm1b = benchmark "ack 3 7 on b" (ack (b (toInteger 3)) (b (toInteger 7)))
bm1n = benchmark "ack 3 7 on n" (ack (n (toInteger 3)) (n (toInteger 7)))

bm2t = benchmark "exp2 t" (exp2 (exp2 (t (toInteger 14))))
bm2b = benchmark "exp2 b" (exp2 (exp2 (b (toInteger 14))))
bm2n = benchmark "exp2 n" (exp2 (exp2 (n (toInteger 14))))

bm3 tvar = benchmark "sparse_set on a type" (n (bitsize (from_set ps)))
  where ps = map tvar [101,2002..100000]

bm4t =benchmark "bitsize of Mersenne 45"  (n (bitsize mersenne45))
bm5t = benchmark "bitsize of Perfect 45" (n (bitsize perfect45))

bm6t = benchmark "large leftshift" (leftshift n n) where
  n = t prime45

bm3' tvar m = benchmark "to/from list on a type" 
              (n (bitsize (from_list (to_list (from_list ps)))))
  where ps = map tvar [101,2002..3000+m]

bm3'' tvar m = benchmark "to/from list on a type" 
              (n (bitsize (from_list (to_list (from_list ps)))))
  where ps = map (dual.tvar) [101,2002..3000+m]

bm7t = benchmark "primes on t" 
  (last (take 100 ps)) where ps = primes :: [T]
bm7b = benchmark "primes on b" 
  (last (take 100 ps)) where ps = primes :: [B]
bm7n = benchmark "primes on n" 
  (last (take 100 ps)) where ps = primes :: [Integer]

bm8t = benchmark "mersenne on t" 
  (last (take 7 ps)) where ps = mersenne_primes :: [T]
bm8b = benchmark "mersenne on b" 
  (last (take 7 ps)) where ps = mersenne_primes :: [B]
bm8n = benchmark "mersenne on n" 
  (last (take 7 ps)) where ps = mersenne_primes :: [Integer]
\end{code}

The following tests the syracuse / Collatz conjecture up to {\tt m}
\begin{code}
test_syr tvar m = maximum (map length (map (nsyr . tvar) [0..m]))

compress_syr tvar m = r where
  nss = map (nsyr . tvar) [0..m]
  r = maximum (map (n.bitsize) (map from_list nss))

-- overflows for m>2 except for tvar=t
compress_syr_twice tvar m = r  where
  nss = map (nsyr . tvar) [0..m]
  r = (n.bitsize) (from_list (map from_list nss))
  
bm9 tvar = benchmark "test syracuse" (test_syr tvar 2000)

bm10 tvar = benchmark "compress syracuse" (compress_syr tvar 100)

bm11 tvar = benchmark "compress syracuse_twice" r where
            r = compress_syr_twice tvar 20
\end{code}

The following function computes the size of a tree-represented
natural number:
\begin{code}
tsize T = 1
tsize (V x xs) = 1+ sum (map tsize (x:xs))
tsize (W x xs) = 1+ sum (map tsize (x:xs))
\end{code}

The function {\tt kth} computes the k-th iteration of a function
application.
\begin{code}
kth _ k x | e_ k  = x
kth f k x  = f (kth f (s' k) x)
\end{code}

The following assertions are used for testing some of our operations:

\begin{code}
-- relation between iterations of o,i and power of 2 
a1 k = pow (i e) k == s (kth o k e) 
a2 k = pow (i e) k == s (s (kth i (s'  k) e))

-- relations between power operations and multiplication
a3 n b = (u==v,u,v) where
 m = pow (i e) n
 u = kth o n b
 v = s' (mul m (s b))
 
a4 x y = (a==b,a,b) where 
  a = mul (pow (i e) x) y
  b = s (kth o x (s' y))
\end{code}

\end{document}